\documentclass[11pt]{article}  

\usepackage{subfigure,enumerate,bm,amsmath,amsthm,wrapfig,array,color,algorithmic,amssymb}
\usepackage{graphics}
\usepackage{natbib}

\newtheorem{theorem}{\bf Theorem}[]
\newtheorem{lemma}[theorem]{Lemma}

\newtheorem{corollary}[theorem]{Corollary}

\newtheorem{thm}[theorem]{\bf Theorem}

\def\math#1{$#1$}

\def\mand#1{$$#1$$}
\def\v#1{{\mathbf #1}}
\def\frac#1#2{{#1\over #2}}

\def\mld#1{\begin{equation}
#1
\end{equation}}
\def\eqar#1{\begin{eqnarray}
#1
\end{eqnarray}}
\def\eqan#1{\begin{eqnarray*}
#1
\end{eqnarray*}}

\DeclareSymbolFont{AMSb}{U}{msb}{m}{n}
\DeclareMathSymbol{\N}{\mathbin}{AMSb}{"4E}
\DeclareMathSymbol{\Z}{\mathbin}{AMSb}{"5A}
\DeclareMathSymbol{\R}{\mathbin}{AMSb}{"52}
\DeclareMathSymbol{\Q}{\mathbin}{AMSb}{"51}
\DeclareMathSymbol{\I}{\mathbin}{AMSb}{"49}
\DeclareMathSymbol{\C}{\mathbin}{AMSb}{"43}

\def\cl#1{{\cal #1}}

\def\x{{\mathbf x}}
\def\y{{\mathbf y}}
\def\z{{\mathbf z}}
\def\X{{\mathbf X}}

\def\a{{\mathbf a}}
\def\b{{\mathbf b}}
\def\exp#1{{\left\langle#1\right\rangle}}
\def\E#1{{{\mathbf E}\left[#1\right]}}
\def\norm#1{{\|#1\|}}

\def\r#1{{(\ref{#1})}}

\def\dotfil{\leaders\hbox to 1.5mm{.}\hfill}
\newcounter{rmnum}
\def\RN#1{\setcounter{rmnum}{#1}\uppercase\expandafter{\romannumeral\value{rmnum}}}
\def\rn#1{\setcounter{rmnum}{#1}\expandafter{\romannumeral\value{rmnum}}}

\newcommand{\transp}{\ensuremath{^\text{\textsc{t}}}}
\newcommand{\trace}{\text{\rm trace}}
\newcommand{\mat}[1]{{\ensuremath{\textsc{#1}}}}

\def\exp{\hbox{\rm exp}}
\def\rank{\hbox{\rm rank}}
\def\col{\hbox{\rm col}}
\def\ker{\hbox{\rm ker}}

\def\b{{\mathbf b}}
\def\e{{\mathbf e}}

\def\rb{{\mathbf r}}
\def\u{{\mathbf u}}
\def\v{{\mathbf v}}
\def\xhat{{\hat\x}}

\def\A{\matA}

\def\Atilde{\tilde\matA}
\def\Btilde{\tilde\matB}
\def\Stilde{\tilde\matS}
\def\Utilde{\tilde\matU}
\def\Vtilde{\tilde\matV}
\def\E{{\cl E}}

\def\Q{{\bm{Q}}}

\def\matA{\mat{A}}
\def\matB{\mat{B}}

\def\matI{\mat{I}}

\def\matQ{\mat{Q}}

\def\matS{\mat{S}}
\def\matU{\mat{U}}
\def\matV{\mat{V}}
\def\matW{\mat{W}}
\def\matX{\mat{X}}

\def\matZ{\mat{Z}}

\def\w{{\mathbf{w}}}

\DeclareMathSymbol{\Prob}{\mathbin}{AMSb}{"50}
\DeclareMathSymbol{\Exp}{\mathbin}{AMSb}{"45}
\newcommand\remove[1]{}

\begin{document}

\title{Using a Non-Commutative Bernstein Bound to Approximate Some
Matrix Algorithms in the Spectral Norm}

\author{Malik Magdon-Ismail\\
CS Department, Rensselaer Polytechnic Institute,\\
Troy, NY 12180, USA. \\
{\sf magdon@cs.rpi.edu}
}

\maketitle
\begin{abstract}%
We focus on \emph{row sampling} based 
approximations for matrix algorithms, in particular 
matrix multipication,
sparse matrix reconstruction, and \math{\ell_2} regression.
For \math{\matA\in\R^{m\times d}} 
(\math{m} points in \math{d\ll m} dimensions), 
and appropriate row-sampling probabilities, which typically depend on the
norms of the rows of the \math{m\times d} left singular matrix of
\math{\matA} (the \emph{leverage scores}),
we give row-sampling algorithms with linear (up to polylog factors)
 dependence on the
stable rank of~\math{\matA}.
This result is achieved through the application of 
non-commutative Bernstein bounds.

\end{abstract}

\section{Introduction}
\label{section:intro}

Matrix algorithms (eg. matrix multiplication, SVD, \math{\ell_2} regression)
 are of widespread use in many application areas:
data mining \citep{azar2001}; recommendations systems \citep{drineas2002};
information retrieval \citep{berry1995,papadimitriou2000}; web search
\citep{klienberg1999,achlioptas2001}; clustering 
\citep{drineas2004,mcsherry2001}; mixture modeling 
\citep{kannan2008,achlioptas2005}; etc. Based on the importance of matrix
algorithms, there has been considerable research
energy expended on breaking the \math{O(md^2)} bound required by
exact SVD methods \citep{golub2}.

Starting with a seminal result of  
\citet{Frieze1998}, a large number of results using non-uniform
sampling to speed up matrix computations have appeared
\citep{achlioptas2003,deshpande2006,deshpande2006b,drineas2006b,drineas2006c,drineas2006d,drineas2006,drineas2006e,rudelson2007,zouzias2010},
some of which give relative error guarantees
\citep{deshpande2006,deshpande2006b,drineas2006,drineas2006e,zouzias2010}. 

Even more recently, \citet{sarlos2006} showed how random projections
or ``sketches'' can be used to perform all these tasks efficiently, obtaining
the first \math{o(md^2)} algorithms when preserving the identity of the rows
themselves are not important. In fact, we will find many of these
techniques, together with those in \citet{ailon2006} essential to 
our algorithm for generating row samples ultimately leading to 
\math{o(md^2)} algorithms based on row-sampling. From now on, 
we focus on row-sampling algorithms.

We start with the basic result of matrix multiplication. All other results
more or less follow from here.
In an independent recent 
work which is developed along the lines of using 
isoperimetric inequalities \citep{rudelson2007} to obtain matrix Chernoff
bounds, \citet{zouzias2010} show that by sampling 
nearly a linear number of rows, it is possible to obtain a relative 
error approximation to matrix multiplication. Specifically, 
let \math{\matA\in\R^{m\times d_1}} and \math{\matB\in\R^{m\times d_2}}.
Then, for \math{r= \Omega(\rho/\epsilon^2\log(d_1+d_2))} (where
\math{\rho} bounds the stable (or ``soft'') rank of \math{\matA} and
\math{\matB} -- see
later), there is
a probability distribution over \math{\cl{I}=
\{1,\ldots,m\}} such that by sampling
\math{r} rows i.i.d. from \math{\cl{I}},
one can construct sketches
\math{\Atilde,\ \Btilde} such that
\math{\Atilde\transp\Btilde\approx\matA\transp\matB}. Specifically, with
constant probability,
\mand{
\norm{\Atilde\transp\Btilde-\matA\transp\matB}_2
\le \epsilon\norm{\matA}_2\norm{\matB}_2.}
The sampling distribution is relatively simple, relying only on the product
of the norms of
the rows in \math{\matA} and \math{\matB}. This result is applied to low rank
matrix reconstruction and \math{\ell_2}-regression where the required sampling
distribution needs knowledge of the SVD of \math{\matA} and \math{\matB}.

Our basic result for matrix multiplication is very similar to this, and we
arrive at it through a different path using a non-commutative Bernstein bound.
Our sampling probabilities are different. In appication of our results
to sparse matrix reconstruction and \math{\ell_2}-regression, the rows
of the left
singular matrix make an appearance. In \cite{malik143}, it is shown
 how to approximate these
 probabilities  in \math{o(md^2)} time
using random projections at the expense of a poly-logarithmic
factor in running times. 
Further refinements lead to an even more efficient algorithm
\cite{malik144}.
As mentioned above, we must confess
that one may perform
our matrix tasks more efficiently using these same random
projection methods \citep{sarlos2006}, 
however the resulting algorithms are in terms of a small
number of linear combinations of all the rows. In many applications, the
actual rows of \math{\matA} have some physical meaning and so methods
based on a small number of the actual rows
are of interest.

We finally mention that \cite{zouzias2010} also give a dimension
independent bound for matrix multiplication using some stronger tools. 
Namely, one can get the matrix multiplication approximation in the
spectral norm using \math{r= \Omega(\rho/\epsilon^2\log(\rho/\epsilon^2))}.
In practice, it is not clear which
bound is better, since there is now an additional factor of 
\math{1/\epsilon^2} inside the logarithm. 

\subsection{Basic Notation}

Before we can state the results in concrete form, we need some preliminary 
conventions. 
In general, \math{\epsilon\in(0,1)} will be an error tolerance
parameter; \math{\beta\in(0,1]} is a parameter
used to scale probabilities; and, \math{c,c'>0} are generic constants
 whose value may vary even
within different lines of the same derivation.
Let \math{\e_1,\ldots,\e_m} be the standard basis vectors
in \math{\R^m}.
Let \math{\matA\in\R^{m\times d}} denote an arbitrary matrix
which represents \math{m} points in \math{\R^d}. 
In general, we might represent a matrix such as \math{\A} (roman, uppercase)
by a set of vectors \math{\a_1,\ldots,\a_m\in\R^d} (bold, lowercase), 
so that \math{\A\transp=[\a_1\ \a_2\ \ldots\ \a_m]};
similarly, for a vector \math{\y}, \math{\y\transp=[y_1,\ldots,y_m]}.
Note that \math{\a_t} is the 
\math{t^{\hbox{th}}} row of \math{\matA}, which we may also
refer to by \math{\matA_{(t)}}; similarly, we may refer to the 
\math{t^{\hbox{th}}} column as \math{\matA^{(t)}}.
Let
\math{\rank(\matA)\le\min\{m,d\}} be  the rank of \math{\matA}; typically
\math{m\gg d} and for concreteness, we will assume
that \math{\rank(\matA)=d} 
(all the results easily generalize to \math{\rank(\matA)<d}). 
For matrices, we will use the spectral norm,
\math{\norm{\cdot}}; on occasion, we will use the Frobenius norm,
\math{\norm{\cdot}_{F}}.
For vectors, \math{\norm{\cdot}_{F}=\norm{\cdot}} 
(the standard Euclidean norm). The stable, or ``soft'' rank,
\math{\rho(\matA)=\norm{\matA}_F^2/\norm{\matA}^2\le\rank(\matA)}.

The singular value decomposition (SVD)
of 
\math{\matA} is 
\mand{
\matA=\matU_\matA\matS_\matA \matV_\matA\transp.
}
where \math{\matU_A} is an \math{m\times d} set of columns which are an
orthonotmal basis for the column space in \math{\matA}; \math{\matS_\matA}
 is a
\math{d\times d} positive diagonal matrix of singular values, and
\math{\matV} is a \math{d\times d} orthogonal matrix. 
We refer to the singular values of \math{\matA} (the diagonal entries
in \math{\matS_\matA}) by
\math{\sigma_i(\matA)}.
We will call a matrix
with orthonormal columns an orthonormal matrix; an orthogonal matrix is
a square orthonormal matrix.
In particular, 
\math{\matU_\matA\transp\matU_\matA=\matV_\matA\transp\matV_\matA=
\matV_\matA\matV_\matA\transp=\matI_{d\times d}}. It is possible to extend 
\math{\matU_\matA} to a full orthonormal basis of \math{\R^m},
\math{[\matU_\matA,\matU_\matA^\perp]}. 

The SVD is important for a number of reasons. The projection of the columns
of \math{\matA} onto the \math{k} left singular vectors with top \math{k}
singular values gives the best rank-\math{k} approximation to \math{\matA} in 
the spectral and Frobenius norms. The solution to the linear
regression problem is also
intimately related to the SVD. In particular, consider 
the following minimization problem which is minimized at \math{\w^*}:
\mand{
Z^*=\min_{\w}\norm{\matA\w-\y}^2.
} 
It is known \citep{golub2} that
\math{Z^*=\norm{\matU_A^\perp(\matU_\matA^\perp)\transp\y}^2}, and
 \math{\w^*=\matV_\matA\matS_\matA^{-1}\matU_\matA\transp\y}.

\paragraph{Row-Sampling Matrices}

Our focus is algorithms based on row-sampling.
A \emph{row-sampling matrix}
 \math{\matQ\in\R^{r\times m}} samples 
\math{r} rows of \math{\matA} to form \math{\Atilde=\matQ\matA}:
\mand{
\matQ=
\left[
\begin{matrix}
\rb_1\transp\\
\vdots\\
\rb_r\transp
\end{matrix}
\right],\qquad
\Atilde=
\matQ\matA=
\left[
\begin{matrix}
\rb_1\transp\matA\\
\vdots\\
\rb_r\transp\matA
\end{matrix}
\right]
=
\left[
\begin{matrix}
\lambda_{t_1}\a_{t_1}\transp\\
\vdots\\
\lambda_{t_r}\a_{t_r}\transp
\end{matrix}
\right]
,
}
where \math{\rb_j=\lambda_{t_j}\e_{t_j}}; it is easy to verify that the row 
\math{\rb_j\transp\matA} samples the \math{t_j^{\hbox{th}}} 
row of \math{\matA} and
rescales it.
We are interested in 
random sampling matrices where each \math{\rb_j} is i.i.d. according to 
some distribution. Define 
a set of sampling probabilities
\math{p_1,\ldots,p_m}, with \math{p_i\ge 0} and \math{\sum_{i=1}^mp_i=1};
then \math{\rb_j=\e_t/\sqrt{rp_t}} with probability \math{p_t}.
Note that the scaling is also related to the sampling probabilities 
in all the algorithms we consider.
\remove{
As an example, if 
\math{m=5} and \math{r=3}, and the sampling probabilities are
\math{p_1,\ldots,p_5}, then an example sampling matrix if 
the sampled indices happen to be \math{\{2,5,3\}} is given below:
\mand{
\matQ_{\{2,5,3\}}=
\frac{1}{\sqrt{3}}\left[
\begin{matrix}
0&\frac{1}{\sqrt{p_2}}&0&0&0\\
0&0&0&0&\frac{1}{\sqrt{p_5}}\\
0&0&\frac{1}{\sqrt{p_3}}&0&0\\
\end{matrix}
\right].
}
The row-sampled version of \math{\matA} using \math{\matQ} is
\math{\matA_r=\matQ\matA}; one can verify that the rows of 
\math{\matA_r} are rescaled versions of the rows in \math{\matA}. Using
\math{\matA_{(t)}} to denote the \math{t^{\hbox{th}}} row in \math{\matA},
for the example sampling matrix above,
\mand{
\matA_r=\matQ_{\{2,5,3\}}\matA=
\frac{1}{\sqrt{3}}
\left[
\begin{matrix}
\frac{1}{\sqrt{p_2}}\matA_{(2)}\\
\frac{1}{\sqrt{p_5}}\matA_{(5)}\\
\frac{1}{\sqrt{p_3}}\matA_{(3)}\\
\end{matrix}
\right].
}
}
We can write \math{\matQ\transp\matQ} as the sum of \math{r} independently
sampled matrices,
\mand{
\matQ\transp\matQ=\frac{1}{r}\sum_{j=1}^r\rb_j\rb_j\transp
}
where \math{\rb_j\rb_j\transp} is a diagonal matrix with only one non-zero
diagonal entry; the \math{t^{th}} diagonal entry is equal
to \math{1/p_t} with probability \math{p_t}. Thus, by construction, for any
set of non-zero sampling probabilities, 
\math{\Exp[\rb_j\rb_j\transp]=\matI_{m\times m}}. Since we are averaging \math{r} independent
copies, it is reasonable to expect a concentration around the mean, with 
respect to \math{r}, and so in some sense, \math{\matQ\transp\matQ} 
essentially behaves like the identity. 

\subsection{Statement of Results}

The two main results relate to how orthonormal subspaces behave with 
respect to the row-sampling. These are discussed more
thoroughly in Section \ref{section:orthonormal}, but we state them here 
summarily.
\begin{thm}[Symmetric Orthonormal Subspace Sampling]
\label{thm:symmetric}
Let \math{\matU\in\R^{m\times d}} be orthonormal, and 
\math{\matS\in\R^{d\times d}} be positive diagonal.
Assume the  row-sampling probabilities \math{p_t}
satisfy
\mand{
p_t\ge\beta\frac{\u_t\transp\matS^2\u_t}{\trace(\matS^2)}.
}
Then, if 
\math{r\ge(4\rho(\matS)/\beta\epsilon^2)\ln\frac{2d}{\delta}}, 
with probability
at least \math{1-\delta},
\mand{
\norm{\matS^2-\matS\matU\transp\matQ\transp\matQ\matU\matS}\le
\epsilon\norm{\matS}^2
}
\end{thm}
We also have an asymmetric version of Theorem \ref{thm:symmetric}, which
is actually obtained through an application of Theorem \ref{thm:symmetric}
to a composite
matrix.
\begin{thm}[Asymmetric Orthonormal Subspace Sampling]
Let \math{\matW\in
\R^{m\times d_1},\ \matV\in\R^{m\times d_2}} be orthonormal, and let 
\math{\matS_1\in\R^{d_1\times d_1}} and
\math{\matS_2\in\R^{d_2\times d_2}} be two positive diagonal matrices;
let \math{\rho_i=\rho(\matS_i)}.
Consider row sampling probabilities
\mand{
p_t\ge\beta\frac{\frac{1}{\norm{\matS_1}^2}\w_t\transp\matS_1^2\w_t+
\frac{1}{\norm{\matS_2}^2}\v_t\transp\matS_2^2\v_t}
{\rho_1+\rho_2}.
}
If \math{r\ge(8(\rho_1+\rho_2)/
\beta\epsilon^2)\ln\frac{2(d_1+d_2)}{\delta}}, then with 
probability at least \math{1-\delta},
\mand{
\norm{\matS_1\matW\transp\matV\matS_2-
\matS_1\matW\transp\matQ\transp\matQ\matV\matS_2}
\le \epsilon\norm{\matS_1}\norm{\matS_2}
}
\end{thm}
We note that these row sampling probabilities are not the usual product
row sampling probabilities one uses for matrix multiplication as in
\citet{drineas2006b}. 
Computing the probabilities  requires knowledge of 
the spectral norms of \math{\matS_i}. Here, 
\math{\matS_i} are given diagonal matrices, so it is easy
to compute \math{\norm{\matS_i}}. In the application of these results
to matrix multiplication, the spectral norm of the input matrices will
appear. We will show how to handle this issue later. 
As a byproduct, we will give an efficient algorithm to obtain a relative
error approximation to \math{\norm{\matA}} based on row sampling and the
power-iteration, which improves upon \cite{woolfe2008,kuczynski1989}.

We now give 
some applications of these orthonormal subspace sampling results.
\begin{thm}[Matrix Multiplication in Spectral Norm]
Let \math{\matA\in\R^{m\times d_1}} and \math{\matB\in\R^{m\times d_2}}
 have
rescaled rows \math{\hat\a_t=\a_t/\norm{\matA}} and \math{\hat\b_t=
\b_t/\norm{\matB}} respectively.  Let \math{\rho_A} (resp. \math{\rho_B})
be the stable rank of \math{\matA} (resp. \math{\matB}).
Obtain a sampling  matrix
\math{\matQ\in\R^{r\times m}} using 
row-sampling probabilities \math{p_t} satisfying
\mand{
p_t\ge\beta \frac{\hat\a_t\transp\hat\a_t+
\hat\b_t\transp\hat\b_t}
{\sum_{t=1}^m\hat\a_t\transp\hat\a_t+
\hat\b_t\transp\hat\b_t}
=\beta  \frac{\hat\a_t\transp\hat\a_t+\hat\b_t\transp\hat\b_t}
{\rho_{\matA}+\rho_{\matB}}.
}
Then, if \math{r\ge \frac{8(\rho_\matA+\rho_\matB)}{\beta\epsilon^2}
\ln\frac{2(d_1+d_2)}{\delta}}, with
probability at least \math{1-\delta},
\mand{
\norm{\matA\transp\matB-\Atilde\transp\Btilde}\le
\epsilon\norm{\matA}\norm{\matB}.
}
\end{thm}

The sampling probabilities depend on \math{\norm{\matA}^2} and
\math{\norm{\matB}^2}.
It is possible to get a constant factor approximation to 
\math{\norm{\matA}^2} (and similarly
\math{\norm{\matB}^2}) with high probability.
We summarize the idea here, the details are given in Section 
\ref{section:sampling}, Theorem \ref{theorem:spectral}.
First sample \math{\Atilde=\matQ\matA} according to probabilities
\math{p_t=\a_t^2/\norm{\matA}_F^2}. These probabilities are easy to compute
in \math{O(md_1)}. 
By an 
application of the symmetric subspace sampling theorem 
(see Theorem~\ref{theorem:matmultsym}),
if
\math{r\ge(4\rho_A/\epsilon^2)\ln\frac{2d_1}{\delta}}, then 
with probability at least \math{1-\delta},
\mand{
(1-\epsilon)\norm{\matA}^2\le\norm{\Atilde\transp\Atilde}
\le(1+\epsilon)\norm{\matA}^2.
}
We now run \math{\Omega(\ln{\frac{d_1}{\delta}})}
 power iterations starting from a random
isotropic vector  to estimate
the spectral norm of \math{\Atilde\transp\Atilde}.
The efficiency is \math{O(md_1+\rho_A d_1/\epsilon^2\ln^2 (\frac{d_1}{\delta})
)}.

\begin{thm}[Sparse Row-Based Matrix Reconstruction]
Let \math{\matA} have the SVD representation
\math{\matA=\matU\matS\matV\transp}, and consider
row-sampling probabilities \math{p_t} satisfying
\math{p_t\ge\frac{\beta}{d}\u_t\transp\u_t}. Then, if 
\math{r\ge(4(d-\beta)/\beta\epsilon^2)\ln\frac{2d}{\delta}}, with probability
at least \math{1-\delta},
\mand{
\norm{\matA-\matA\tilde\Pi_k}\le
\left(\frac{1+\epsilon}{1-\epsilon}\right)^{1/2}
\norm{\matA-\matA_k},
}
for \math{k=1,\ldots,d},
where \math{\tilde\Pi_k} projects onto the top \math{k} right singular vectors
of \math{\Atilde}.
\end{thm}

It is possible to obtain relative approximations to
the sampling probabilities according to 
the rows of the left singluar matrix (the leverage scores),
 but that goes beyond the scope of this
work \cite{malik143,malik144}

\begin{thm}[Relative Error \math{\ell_2} Regression]
Let \math{\matA\in\R^{m\times d}} have the SVD representation
\math{\matA=\matU\matS\matV\transp}, and let \math{\y\in\R^m}.
Let \math{\x^*=\matA^+\y} be the optimal regression with 
residual \math{\bm\epsilon=\y-\matA\x^*=\y-\matA\matA^+\y}.
Assume the sampling probabilities \math{p_t} satisfy
\mand{
p_t\ge
{\beta}\left(\frac{\u_t^2}{d}+
\frac{
(\u_t^2+
\frac{\epsilon_t^2}{\bm\epsilon\transp\bm\epsilon})}{d+1}+
\frac{\epsilon_t^2}{\bm\epsilon\transp\bm\epsilon}
\right)
}
For  
\math{r\ge(8(d+1)/\beta\epsilon^2)\ln\frac{2(d+1)}{\delta}}, let
\math{\hat\x=(\matQ\matA)^+\matQ\y} be the approximate regression. Then,
with probability at least \math{1-3\delta},
\mand{
\norm{\matA\hat\x-\y}\le \left(1+
\epsilon+\epsilon\sqrt{\frac{1+\epsilon}{1-\epsilon}}
\right)
\norm{\matA\x^*-\y}.
}
\end{thm}
In addition to sampling according to \math{\u_t^2} we also need the residual
vector \math{\bm\epsilon=\y-\matA\matA^+\y}. Unfortunately, we have not
yet
found an efficient way to get a good approximation (in some
form of relative error) to this residual vector.

\subsection{Paper Outline}
Next we describe some probabistic tail inequalities which will be useful.
We continue with the sampling lemmas for orthonormal matrices, followed by the 
applications to matrix multiplication, matrix reconstruction and
\math{\ell_2}-regression. Finally, we discuss the algorithm
for approximating the spectral norm based on sampling and the power iteration.

\section{Probabilistic Tail Inequalities}
Since all our arguments involve high probability results, our main bounding
tools will be probability tail inequalities. First, let \math{X_1,\ldots,
X_n} be independent random variables with \math{\Exp[X_i]=0} and 
\math{|\X_i|\le \gamma}; let \math{Z_n=\frac{1}{n}\sum_{i=1}^n X_i}.
Chernoff, and later Hoeffding gave the bound
\begin{theorem}[\citet{chernoff1952,hoeffding1963}]
\label{theorem:chernoff}
\math{\displaystyle \Prob[|Z_n|>\epsilon]\le 2e^{-n\epsilon^2/2\gamma^2}}.
\end{theorem}
If in addition one can bound the variance,
\math{\Exp[X_i^2]\le s^2}, then we have Bernstein's bound:
\begin{theorem}[\citet{bernstein1924}]
\label{theorem:bernstein}
\math{\displaystyle
\Prob[|Z_n|\ge\epsilon]\le 2e^{-n\epsilon^2/(2s^2+2\gamma\epsilon/3)}.
}
\end{theorem}
Note that when \math{\epsilon\le3s^2/\gamma}, we can simplify the Bernstein
bound to \math{\Prob[|Z_n|\ge\epsilon]\le 2e^{-n\epsilon^2/4s^2}}, which is
considerably simpler and only involves the variance.
The non-commutative versions of these bounds, which extend these inequalities
to matrix valued random variables can also be deduced. Let
\math{\matX_1,\ldots,
\matX_n} be independent copies of a symmetric random matrix
\math{\matX},  with \math{\Exp[\matX]=\bm 0}, and 
suppose that \math{\norm{\matX}_2\le\gamma};
let \math{\matZ_n=\frac{1}{n}\sum_{i=1}^n\matX_i}.
 \citet{ahlswede2002}
gave the fundamental extension of the exponentiation trick for computing
Chernoff bounds of scalar random variables to matrix 
valued random variables (for a simplified proof, see \citet{wigderson2008}):
\mld{
\Prob[
\norm{\matZ_n}_2>\epsilon]
\le\inf_t
2de^{-n\epsilon t/\gamma}
\norm{\Exp[e^{t\matX/\gamma}]}_2^n.
\label{eq:ahlswede}
}
By standard optimization of this bound, 
one readily obtains the non-commutative tail inequality
\begin{theorem}[\citet{ahlswede2002}]
\label{theorem:ahlswede}
\math{\displaystyle
\Prob[
\norm{\matZ_n}_2>\epsilon]\le
2de^{-n\epsilon^2/4\gamma^2}.
}
\end{theorem}
\begin{proof}
The statement is trivial if \math{\epsilon\ge\gamma}, so assume
\math{\epsilon<\gamma}. The lemma follows from 
\r{eq:ahlswede} and the following sequence after
setting \math{t=\epsilon/2\gamma\le\frac{1}{2}}:
\mld{
\norm{\Exp[e^{t\matX/\gamma}]}_2
\mathop{\buildrel{(a)}\over{\le}}
1+\sum_{\ell=2}^\infty\frac{t^\ell}{\ell!}\Exp[\norm{(\matX/\gamma)^\ell}_2]
\mathop{\buildrel{(b)}\over{\le}}
1+t^2
\mathop{\le}
e^{t^2},\label{eq:ahlswede_proof}
}
where (a) follows from  \math{\Exp[\matX]=0}, the triangle inequality and
\math{\norm{\Exp[\cdot]}_2\le\Exp[\norm{\cdot}_2]};
(b) follows because \math{\norm{\matX/\gamma}_2\le 1} and \math{t\le\frac12}. 
\end{proof}
(We have stated a simplified version of the bound, without 
taking care to optimize the constants.) As mentioned in 
\citet{gross2009}, one can obtain a non-commuting version of Bernstein's
inequality in a similar fashion using \r{eq:ahlswede}. Assume that 
\math{\norm{\Exp\matX\transp\matX}_2\le s^2}. 
By adapting the standard Bernstein bounding argument to matrices, we have
\begin{lemma}\label{lemma:bernstein_trick}
For symmetric \math{\matX},
\math{\displaystyle 
\norm{\Exp[e^{t\matX/\gamma}]}_2
\le
\exp\left({\textstyle\frac{s^2}{\gamma^2}}(e^t-1-t)
\right).
}
\end{lemma}
\begin{proof}
As in \r{eq:ahlswede_proof}, but using 
submultiplicativity, we first bound
\math{\norm{\Exp[\matX^\ell]}_2\le s^2\gamma^{\ell-2}}:
\eqan{
\norm{\Exp[\matX^\ell]}_2
&=&
\max_{\norm{\u}=1}\left\|\int d\matX\ p(\matX)\matX^\ell\u\right\|\\
&=&
\max_{\norm{\u}=1}
\left\|\int d\matX\ p(\matX)
\frac{\norm{\matX^{\ell-2}\u}\matX^2\matX^{\ell-2}\u}{\norm{\matX^{\ell-2}\u}}
\right\|\\
&\le&
\gamma^{\ell-2}
\max_{\norm{\w}=1}
\left\|\int d\matX\ p(\matX)
\matX^2\w
\right\|\\
&=&
\gamma^{\ell-2}\norm{\Exp[\matX^2]}_2\le s^2\gamma^{\ell-2}.
}
To conclude, we use the triangle inequality to bound as follows:
\mand{
\norm{\Exp[e^{t\matX/\gamma}]}_2
=
\left\|
\matI+\sum_{\ell=2}^\infty\frac{t^\ell}{\gamma^\ell\ell!}\Exp[\matX^\ell]
\right\|_2
\le
1+\frac{s^2}{\gamma^2}\sum_{\ell=2}^\infty \frac{t^\ell}{\ell!}=
1+\frac{s^2}{\gamma^2}(e^t-1-t)\le \exp
\left({\textstyle\frac{s^2}{\gamma^2}}(e^t-1-t)
\right).
}
\end{proof}
Using Lemma \ref{lemma:bernstein_trick} in \r{eq:ahlswede} with
\math{t=\ln(1+\epsilon\gamma/s^2)}, and using 
\math{(1+x)\ln(1+\frac{1}{x}) -1\ge \frac{1}{2x+2/3}}, we obtain the 
following result.
\begin{theorem}[Non-commutative Bernstein]
\label{theorem:matBernstein}
\math{\displaystyle
\Prob[
\norm{\matZ_n}_2>\epsilon]\le
2de^{-n\epsilon^2/(2s^2+2\gamma\epsilon/3)}.
}
\end{theorem}
\citet{gross2009} gives a simpler version of this
 non-commutative Bernstein
inequality. If \math{\matX\in\R^{d_1\times d_2}} is not symmetric, 
then by considering 
\mand{
\left[
\begin{matrix}
\bm 0_{d_1\times d_1}&\matX\\
\matX\transp&\bm0_{d_2\times d_2}
\end{matrix}
\right],
}
one can get a non-symmetric verision of the non-commutative
Chernoff and  Bernstein bounds,
\begin{theorem}[\cite{recht2009}]
\math{\displaystyle
\Prob[
\norm{\matZ_n}_2>\epsilon]\le
(d_1+d_2)e^{-n\epsilon^2/(2s^2+2\gamma\epsilon/3)}.
}
\end{theorem}
For most of our purposes, we will only need the 
symmetric version; again, if \math{\epsilon\le3s^2/\gamma}, then
we have the much simpler bound 
\math{\Prob[
\norm{\matZ_n}_2>\epsilon]\le 2de^{-n\epsilon^2/4s^2}.}

\section{Orthonormal Sampling Lemmas}
\label{section:orthonormal}

Let \math{\matU\in\R^{m\times d}} be an orthonormal matrix, and let 
\math{\matS\in\R^{d\times d}} be a diagonal matrix. We are interested
in the product \math{\matU\matS\in\R^{m\times d}}; \math{\matU\matS} is the
matrix with columns \math{\matU^{(i)}S_{ii}}. Without loss of 
generality, we can assume that \math{\matS} is positive by flipping the 
signs of the appropriate columns of \math{\matU}. The row-representation
of \math{\matU} is \math{\matU\transp=[\u_1,\ldots,\u_m]}; 
we consider the row
sampling probabilities 
\mld{
p_t\ge\beta\frac{\u_t\transp\matS^2\u_t}{\trace({\matS^2})}.
\label{eq:ptsym}
}
Since \math{\matU\transp\matU=\matI_{d\times d}}, one can verify that
\math{\trace(\matS^2)=\sum_t\u_t\transp\matS^2\u_t} is the correct 
normalization.
\begin{lemma}[Symmetric Subspace Sampling Lemma]
\label{lemma:main1}
\eqan{
\Prob[\norm{\matS^2-\matS\matU\transp\matQ\transp\matQ\matU\matS}
>\epsilon\norm{\matS}^2]
&\le& 
2d\cdot\exp
\left(\frac{-r\epsilon^2}{2(\rho/\beta-\kappa^{-4}
+\epsilon(\rho/\beta-\kappa^{-2})/3)}\right),\\
&\le&
2d\cdot\exp
\left(\frac{-r\beta\epsilon^2}{4\rho}\right),\\
}
where \math{\rho} is the numerical (stable) rank of \math{\matS},
\math{\rho(\matS)=\norm{\matS}^2_F/\norm{\matS}^2}, and
\math{\kappa(\matS)=\sigma_{\max}(\matS)/\sigma_{\min}(\matS)} 
is the condition number.
\end{lemma}
\paragraph{Remarks.}
The stable rank \math{\rho\le d} measures the 
effective dimension of the matrix. 
The condition number \math{\kappa\ge1}, 
hence the simpler version of the bound,
which is valid for 
\math{\epsilon\le 3}. It immediately follows that if 
\math{r\ge(4\rho/\beta\epsilon^2)\ln\frac{2d}{\delta}}, then with probability
at least \math{1-\delta},
\mand{
\norm{\matS^2-\matS\matU\transp\matQ\transp\matQ\matU\matS}\le
\epsilon\norm{\matS}^2
}
An important special case is when \math{\matS=\matI_{d\times d}},
in which case \math{\rho=d}, \math{\kappa=1} and \math{\norm{\matS}=1}.
\begin{corollary}\label{cor:main1}
For sampling probabilities \math{p_t\ge\frac{\beta}{d}\u_t\transp\u_t},
\mand{
\Prob[\norm{\matI-\matU\transp\matQ\transp\matQ\matU}
>\epsilon]\le 2d\cdot\exp
\left(\frac{-\beta r\epsilon^2}{4(d-\beta)}\right),
}
\end{corollary}
\begin{proof}(of Lemma \ref{lemma:main1})
Note that 
\math{\matU\transp\matQ\transp\matQ\matU=\frac{1}{r}\sum_{i=1}^r
\u_{t_i}\u_{t_i}\transp/p_{t_i}}, where \math{t_i\in[1,m]} is chosen
according to the probability \math{p_{t_i}}. It follows that 
\mand{
\matS^2-\matS\matU\transp\matQ\transp\matQ\matU\matS=
\frac{1}{r}\sum_{i=1}^r\matS^2-\frac{1}{p_{t_i}}\matS\u_{t_i}\u_{t_i}\transp
\matS
=\frac{1}{r}\sum_{i=1}\matX_i,
}
where \math{\matX_i} are independent copies of a matrix-random variable
\math{\matX\sim\matS^2-\matS\u\u\transp\matS/p}.
We prove the following three claims:
\eqan{
&&(\rn{1})\ \Exp[\matX]=\bm0;\\
&&(\rn{2})\ \norm{\matX}\le\norm{\matS}^2(\rho/\beta-\kappa^{-2});\\
&&(\rn{2})\ \norm{\Exp\matX\transp\matX}
\le\norm{\matS}^4(\rho/\beta-\kappa^{-4}).
}
The Lemma follows from
the non-commutative Bernstein bound with
\math{\epsilon} replaced by \math{\epsilon\norm{\matS}^2}.
To prove \math{(\rn1)}, note that
\math{\Exp[\matX]=\matS^2-\matS\Exp[\u\u\transp/p]\matS=
\matS^2-\matS\left(\sum_{t=1}^m\u_t\u_t\transp\right)\matS
=\bm0}, because \math{\sum_{t=1}^m\u_t\u_t\transp
=\matU\transp\matU=\matI_{d\times d}}.

To prove \math{(\rn{2})}, let \math{\z} be an arbitrary unit vector and 
consider
\mand{
\z\transp\matX\z=\z\transp\matS^2\z-\frac{1}{p}(\z\transp\matS\u)^2.
}
It follows that \math{\z\transp\matX\z\le\norm{\matS}^2}. To get a lower
bound, we use 
\math{p\ge\beta\u\transp\matS^2\u/\trace({\matS^2})}:
\eqan{
\z\transp\matX\z
&\ge&
\z\transp\matS^2\z-\frac{\trace({\matS^2})}{\beta}
\frac{(\z\transp\matS\u)^2}{\u\transp\matS^2\u},\\
&\buildrel{(a)}\over{\ge}&
\norm{\matS}^2\left(\frac{\sigma_{\min}^2(\matS)}{\norm{\matS}^2}-
\frac{\trace({\matS^2})}{\beta\norm{S}^2}\right),\\
&=&
\norm{\matS}^2\left(\frac{1}{\kappa^2}-
\frac{\rho}{\beta}\right).
}
(a) follows because: by definition of \math{\sigma_{\min}},
the minimum of the first term is
\math{\sigma_{\min}^2}; and, by Cauchy-Schwarz,
\math{(\z\transp\matS\u)^2\le(\z\transp\z)(\u\transp\matS^2\u)}.
Since \math{\beta\le 1}, \math{\rho/\beta-\kappa^{-2}\ge1} (for \math{d>1}), 
and so
\math{|\z\transp\matX\z|\le\norm{\matS}^2\left(\rho/\beta-\kappa^{-2}\right)},
from which \math{(\rn2)} follows.

To prove \math{(\rn{3})}, first note that 
\eqan{
\Exp[\matX\transp\matX]
&=&
\matS^4-\matS^3\Exp[\u\u\transp/p]\matS-\matS\Exp[\u\u\transp/p]\matS^3
+\matS\Exp[\u\u\transp\matS^2\u\u\transp/p^2]\matS,\\
&\buildrel{(a)}\over{=}&
\matS\left(\sum_{t=1}^m
\frac{1}{p_t}\u_t\u_t\transp\matS^2\u_t\u_t\transp\right)\matS
-\matS^4.
}
(a) follows because 
\math{\Exp[\u\u\transp/p]=\matI}.
Thus, for an arbitrary unit \math{\z}, we have
\eqan{
\z\transp\Exp[\matX\transp\matX]\z
&=&
\sum_{t=1}^m\frac{1}{p_t}
(\z\transp\matS\u_t\u_t\transp\matS\z)\u_t\transp\matS^2\u_t-
\z\transp\matS^4\z\transp,\\
&\buildrel{(a)}\over{\le}&
\frac{\trace({\matS^2})}{\beta}
\z\transp\matS\left(\sum_{t=1}^m\u_t\u_t\transp\right)\matS\z
-\z\transp\matS^4\z\transp,\\
&\buildrel{(b)}\over{=}&
\norm{\matS}^4\left(\frac{\trace({\matS^2})}{\beta\norm{\matS}^2}
\frac{\z\transp\matS^2\z\transp}{\norm{\matS}^2}-
\frac{\z\transp\matS^4\z\transp}{\norm{\matS}^4}
\right),\\
&\le&
\norm{\matS}^4\left(\frac{\trace({\matS^2})}{\beta\norm{\matS}^2}-
\frac{\sigma_{\min}^4}{\norm{\matS}^4}
\right).
}
(a) follows from \math{p_t\ge \beta\u_t\transp\matS^2\u_t/\trace(\matS^2)};
(b) follows from 
\math{\matU\transp\matU=\sum_{t=1}^m\u_t\u_t\transp=\matI_{d\times d}}.
Thus, \math{|\z\transp\Exp[\matX\transp\matX]\z|\le
\norm{\matS}^4(\rho/\beta-\kappa^{-4})}, from which \math{(\rn3)} follows.
\end{proof}

For the general case, consider two orthonormal matrices \math{\matW\in
\R^{m\times d_1},\ \matV\in\R^{m\times d_2}},
and two positive diagonal matrices \math{\matS_1\in\R^{d_1\times d_1}} and
\math{\matS_2\in\R^{d_2\times d_2}}. We consider 
the product \math{\matS_1\matW\transp\matV\matS_2}, which is approximated
by the sampled product 
\math{\matS_1\matW\transp\matQ\transp\matQ\matV\matS_2}. Consider
the sampling
probabilities
\mand{
p_t\ge\beta\frac{(\u_t\transp\matS_1^2\u_t)^{1/2}
(\v_t\transp\matS_2^2\v_t)^{1/2}}{
\sum_{t=1}^m(\u_t\transp\matS_1^2\u_t)^{1/2}(\v_t\transp\matS_2^2\v_t)^{1/2}}
\ge
\beta\frac{(\u_t\transp\matS_1^2\u_t)^{1/2}
(\v_t\transp\matS_2^2\v_t)^{1/2}}{\sqrt{\trace({\matS_1^2})\trace({\matS_2^2})}},
}
where the last inequality follows from Cauchy-Schwarz. 
Since \math{\norm{\matA}_F=\sqrt{\rho(\matA)}\norm{\matA}\ge\norm{\matA}},
any bound for the Frobenius norm can be converted into a bound
for the spectral norm. 
Using the Frobenius norm bounds in 
\citet{drineas2006b} (using a simplified form for the bound), one
immediately has:
\mld{
\Prob[\norm{\matS_1\matW\transp\matV\matS_2-
\matS_1\matW\transp\matQ\transp\matQ\matV\matS_2}
>\epsilon\norm{\matS_1}\norm{\matS_2}]\le 
\exp
\left(\frac{-r\beta^2\epsilon^2}{16\rho_1\rho_2}\right),
\label{lemma:drineas1}
}
where \math{\rho_1=\rho(\matS_1)} and \math{\rho_2=\rho(\matS_2)}.
Alternatively, if \math{r\ge (16\rho_1\rho_2/\beta^2\epsilon^2)
\ln\frac{1}{\delta}}, then 
\mand{
\norm{\matS_1\matW\transp\matV\matS_2-
\matS_1\matW\transp\matQ\transp\matQ\matV\matS_2}
\le
\epsilon\norm{\matS_1}\norm{\matS_2}.
}
The dependence on the stable ranks and \math{\beta} is quadratic.
Applying this bound to the 
situation in Lemma \ref{lemma:main1} would give an inferior bound. 
The intuition behind the improvement is that the sampling is isotropic,
and so will not favor any particular direction. One can therefore guess that
all the singular values are approximately equal and so the Frobenius norm
bound on the spectral norm will be loose by a factor of \math{\sqrt{\rho}}; 
and, indeed this is what comes out in the closer analysis.
As a application of Lemma \ref{lemma:main1}, we can get a better
 result for the
asymmetric case.
\begin{lemma}
\label{lemma:main2}
Let \math{\matW\in
\R^{m\times d_1},\ \matV\in\R^{m\times d_2}} be orthonormal, and let 
\math{\matS_1\in\R^{d_1\times d_1}} and
\math{\matS_2\in\R^{d_2\times d_2}} be two positive diagonal matrices.
Consider row sampling probabilities
\mand{
p_t\ge\beta\frac{\frac{1}{\norm{\matS_1}^2}\w_t\transp\matS_1^2\w_t+
\frac{1}{\norm{\matS_2}^2}\v_t\transp\matS_2^2\v_t}
{\rho_1+\rho_2}.
}
If \math{r\ge(8(\rho_1+\rho_2)/
\beta\epsilon^2)\ln\frac{2(d_1+d_2)}{\delta}}, then with 
probability at least \math{1-\delta},
\mand{
\norm{\matS_1\matW\transp\matV\matS_2-
\matS_1\matW\transp\matQ\transp\matQ\matV\matS_2}
\le \epsilon\norm{\matS_1}\norm{\matS_2}
}
\end{lemma}
For the special case that \math{\matS_1=\matI_{d_1\times d_1}} and 
\math{\matS_2=\matI_{d_2\times d_2}}, the sampling probabilities simplify to 
\mand{
p_t\ge\beta\frac{\w_t\transp\w_t+
\v_t\transp\v_t}
{d_1+d_2},
}
\begin{corollary}
\label{cor:main2}
If \math{r\ge(8(d_1+d_2)/
\beta\epsilon^2)\ln\frac{2(d_1+d_2)}{\delta}}, then with probability
at least \math{1-\delta},
\mand{
\norm{\matW\transp\matV-
\matW\transp\matQ\transp\matQ\matV}
\le \epsilon.
}
\end{corollary}
\begin{proof}(of Lemma \ref{lemma:main2})
By homogeneity, we can without loss of generality
assume that \math{\norm{\matS_1}=\norm{\matS_2}=1},
and let\footnote{The general case would have been
 \math{\matZ=\left[\frac{1}{\norm{\matS_1}}\matW\matS_1\ 
\frac{1}{\norm{\matS_2}}\matV\matS_2\right]}.
}
\math{\matZ=[\matW\matS_1\ \matV\matS_2]}. 
An elementary lemma which we will find useful is
\begin{lemma}\label{lemma:basic1}
For any matrix \math{\matA=[\matA_1\ \matA_2]},
\mand{
\max\{\norm{\matA_1},\norm{\matA_2}\}\le
\norm{\matA}\le\sqrt{\norm{\matA_1}^2+\norm{\matA_2}^2}.}
\end{lemma}
The left inequality is saturated when \math{\matA_1} and \math{\matA_2} are 
orthogonal (\math{\matA_1\transp\matA_2=\bm0}), 
and the right inequality is
saturated when \math{\matA_1=\matA_2}. By repeatedly applying 
Lemma \ref{lemma:basic1} one can see that \math{\norm{\matA}} is
at least the spectral norm of any submatrix.
Introduce the SVD of \math{\matZ},
\mand{
\matZ=[\matW\matS_1\ \matV\matS_2]=\matU\matS\matV_\matZ\transp.
}
We now use the
row sampling probabilities according to \math{\matU\matS} from 
\r{eq:ptsym},
\mand{
p_t\ge\beta\frac{\u_t\transp\matS^2\u_t}{\trace({\matS^2})}.
}
We may interpret the sampling probabilities as follows. Let 
\math{\z_t} be a row of \math{\matZ}, the concatenation of 
two rows in \math{\matW\matS_1} and \math{\matV\matS_2}:
\math{\z_t\transp=[\w_t\transp\matS_1\ \v_t\transp\matS_2]}.
We also have that 
\math{\z_t\transp=\u_t\transp\matS\matV_\matZ\transp}. Hence, 
\mand{
\u_t\transp\matS^2\u_t=\u_t\transp\matS\matV_\matZ\transp\matV_\matZ\matS\u_t
=\z_t\transp\z_t=
\w_t\transp\matS_1^2\w_t+\v_t\transp\matS_2^2\v_t.
}
These are exactly the probabilities as claimed in the
statement of the lemma (modulo the rescaling).

Applying Lemma \ref{lemma:main1}: if 
\math{r\ge(4\rho/\beta\epsilon^2)\ln\frac{2\cdot\rank(\matU)}{\delta}}, 
then with 
probability at least \math{1-\delta},
\mand{
\norm{\matS^2-\matS\matU\transp\matQ\transp\matQ\matU\matS}\le
\epsilon\norm{\matS}^2\le\epsilon\sqrt{\norm{\matS_1}^2+\norm{\matS_2}^2}
=\epsilon\sqrt{2},
}
where the second inequality follows from Lemma \ref{lemma:basic1}. 
Since \math{\matZ\matV=\matU\matS},
\mand{
\norm{\matZ\transp\matZ-\matZ\transp\matQ\transp\matQ\matZ}
=\norm{\matS^2-\matS\matU\transp\matQ\transp\matQ\matU\matS}.
}
Further, by the construction of \math{\matZ}, 
\mand{
\matZ\transp\matZ-\matZ\transp\matQ\transp\matQ\matZ
=
\left[
\begin{matrix}
\matS_1^2-
\matS_1\matW\transp\matQ\transp\matQ\matW\matS_1
&\matS_1\matW\transp\matV\matS_2-
\matS_1\matW\transp\matQ\transp\matQ\matV\matS_2
\\
\matS_2\matV\transp\matW\matS_1-
\matS_2\matV\transp\matQ\transp\matQ\matW\matS_1
&\matS_2^2-\matS_2\matV\transp\matQ\transp\matQ\matV\matS_2
\end{matrix}
\right].
}
By Lemma \ref{lemma:basic1}, 
\math{\norm{\matS_1\matW\transp\matV\matS_2-
\matS_1\matW\transp\matQ\transp\matQ\matV\matS_2}\le
\norm{\matZ\transp\matZ-\matZ\transp\matQ\transp\matQ\matZ}
}, and so:
\mand{
\norm{\matS_1\matW\transp\matV\matS_2-
\matS_1\matW\transp\matQ\transp\matQ\matV\matS_2}
\le \epsilon\sqrt{2}.
}
Observe that 
\math{\trace({\matS^2})=\norm{\matZ}^2_F=\trace({\matS_1^2})+\trace({\matS_2^2})};
further,
since \math{\norm{\matS}\ge\max\{\norm{\matS_1},\norm{\matS_2}\}},
we have that
\mand{
\rho(\matS)=\frac{\trace({\matS^2})}{\norm{\matS}^2}
=\frac{\trace({\matS_1^2})+\trace({\matS_2^2})}{\norm{\matS}^2}\le
\frac{\trace({\matS_1^2})}{\norm{\matS_1}^2}+
\frac{\trace({\matS_2^2})}{\norm{\matS_2}^2}=\rho_1+\rho_2.} 
Since
\math{\rank(\matU)\le d_1+d_2}, it 
suffices that \math{r\ge (4(\rho_1+\rho_2)/\beta\epsilon^2)
\ln\frac{2(d_1+d_2)}{\delta}} to obtain error \math{\epsilon\sqrt{2}};
 after rescaling
\math{\epsilon'=\epsilon\sqrt{2}}, we have the result.
\end{proof}

\remove{
\paragraph{Applications}
We will discuss 4 applications of the lemmas in this
section: 
\begin{enumerate}[i.]
\item Sampling approximations to matrix multiplication.
\item Reconstructing a matrix from a sampling of its rows.
\item Sampling for \math{\ell_2} regression with relative error guarantees.
\item A sampling approach to \math{K}-means clustering.
\end{enumerate}
}

\section{Sampling for Matrix Multiplication}

We obtain results for matrix multiplication directly from
Lemmas \ref{lemma:main1} and \ref{lemma:main2}.
First we consider the symmetric case, then the asymmetric case.
Let \math{\matA\in\R^{m\times d_1}} and \math{\matB\in\R^{m\times d_2}}.
We are interested in conditions on the sampling matrix
\math{\matQ\in\R^{r\times m}} such that
\math{\matA\transp\matA\approx\Atilde\transp\Atilde} and 
\math{\matA\transp\matB\approx\Atilde\transp\Btilde}, where
\math{\Atilde=\matQ\matA} and \math{\Btilde=\matQ\matB}.
Using the SVD of \math{\matA},
\eqan{
\norm{\matA\transp\matA-\matA\transp\matQ\transp\matQ\matA}
&=&
\norm{\matV_\matA\matS_\matA\matU_\matA\transp\matU_\matA\matS_\matA
\matV_\matA\transp-\matV_\matA\matS_\matA\matU_\matA\transp
\matQ\transp\matQ\matU_\matA\matS_\matA
\matV_\matA\transp},\\
&=&
\norm{\matS_\matA^2-\matS_\matA\matU_\matA\transp
\matQ\transp\matQ\matU_\matA\matS_\matA}.
}
We may now directly apply Lemma \ref{lemma:main1},
with respect to the appropriate sampling probabilities. 
One can verify that the sampling probabilities in 
Lemma \ref{lemma:main1} are proportional to the
squared norms of the rows of \math{\matA}.
\begin{theorem}
\label{theorem:matmultsym}
Let \math{\matA\in\R^{m\times d_1}}
 have rows \math{\a_t}
Obtain a sampling  matrix
\math{\matQ\in\R^{r\times m}} using 
row-sampling probabilities
\mand{
p_t\ge\beta \frac{\a_t\transp\hat\a_t}
{\norm{\matA}_F^2}.
}
Then, if \math{r\ge \frac{4\rho_\matA}{\beta\epsilon^2}
\ln\frac{2d_1}{\delta}}, with
probability at least \math{1-\delta},
\mand{
\norm{\matA\transp\matA-\Atilde\transp\Atilde}\le
\epsilon\norm{\matA}^2.
}
\end{theorem}
Similarly, using the SVDs of \math{\matA} and \math{\matB},
\eqan{
\norm{\matA\transp\matB-\matA\transp\matQ\transp\matQ\matB}
&=&
\norm{\matV_\matA\matS_\matA\matU_\matA\transp\matU_\matB\matS_\matB
\matV_\matB\transp-\matV_\matA\matS_\matA\matU_\matA\transp
\matQ\transp\matQ\matU_\matB\matS_\matB
\matV_\matB\transp},\\
&=&
\norm{\matS_\matA
\matU_\matA\transp\matU_\matB\matS_\matB-\matS_\matA\matU_\matA\transp
\matQ\transp\matQ\matU_\matB\matS_\matB}.
}
We may now directly apply Lemma \ref{lemma:main2},
with respect to the appropriate sampling probabilities. 
One can verify that the sampling probabilities in 
Lemma \ref{lemma:main2} are proportional to the sum of the
rescaled squared norms of the rows of \math{\matA} and \math{\matB}.
\begin{theorem}
\label{theorem:matmult}
Let \math{\matA\in\R^{m\times d_1}} and \math{\matB\in\R^{m\times d_2}},
 have
rescaled rows \math{\hat\a_t=\a_t/\norm{\matA}} and \math{\hat\b_t=
\b_t/\norm{\matB}} respectively.  
Obtain a sampling  matrix
\math{\matQ\in\R^{r\times m}} using 
row-sampling probabilities
\mand{
p_t\ge\beta \frac{\hat\a_t\transp\hat\a_t+
\hat\b_t\transp\hat\b_t}
{\sum_{t=1}^m\hat\a_t\transp\hat\a_t+
\hat\b_t\transp\hat\b_t}
=\beta  \frac{\hat\a_t\transp\hat\a_t+\hat\b_t\transp\hat\b_t}
{\rho_{\matA}+\rho_{\matB}}.
}
Then, if \math{r\ge \frac{8(\rho_\matA+\rho_\matB)}{\beta\epsilon^2}
\ln\frac{2(d_1+d_2)}{\delta}}, with
probability at least \math{1-\delta},
\mand{
\norm{\matA\transp\matB-\Atilde\transp\Btilde}\le
\epsilon\norm{\matA}\norm{\matB}.
}
\end{theorem}

\section{Sparse Row Based Matrix Representation}

Given a matrix \math{\matA=\matU\matS\matV\transp\in\R^{m\times d}}, 
the top \math{k}
singular vectors, corresponding to the top \math{k} singular values give
the best rank \math{k} reconstruction of \math{\matA}. Specifically,
let \math{\matA_k=\matU_k\matS_k\matV_k\transp}, where
\math{\matU_k\in\R^{m\times k}}, \math{\matS_k\in\R^{k\times k}}
and \math{\matV_k\in\R^{d\times k}}; \math{\matU_k} and \math{\matV_k} 
correspond
to the top-\math{k} left and right singular vectors. Then,
\math{\norm{\matA-\matA_k}\le\norm{\matA-\matX}} where 
\math{\matX\in\R^{m\times d}} ranges over all rank-\math{k} matrices.
As usual, let \math{\Atilde=\matQ\matA} be the sampled, rescaled rows
of \math{\matA}, with \math{\Atilde=\Utilde\Stilde\Vtilde\transp}, and
consider the top-\math{k} right singular vectors 
\math{\Vtilde_k}. Let \math{\tilde\Pi_k} be the projection onto this
top-\math{k} right singular  space, and consider the
rank \math{k} approximation to \math{\matA} obtained by projecting
onto this space:
\math{\Atilde_k=\matA\tilde\Pi_k}. The following lemma is useful for 
showing that \math{\Atilde_k} is almost (up to additive error) as good an
approximation to \math{\matA} as one can get.
\begin{lemma}[\citet{drineas2006c}, \citet{rudelson2007}]
\label{lemma:drineas2}
\mand{
\norm{\matA-\Atilde_k}^2
\le\norm{\matA-\matA_k}^2+2\norm{\matA\transp\matA-\Atilde\transp\Atilde}
\le(\norm{\matA-\matA_k}+
\sqrt{2}\norm{\matA\transp\matA-\Atilde\transp\Atilde}^{1/2})^2.
}
\end{lemma}
\begin{proof}
The proof follows using standard arguments and an application of
a perturbation theory result due to Weyl for bounding the change
in any singular value upon hermitian perturbation of a   
hermitian matrix. 
\end{proof}
Therefore, if we can approximate the matrix product 
\math{\matA\transp\matA}, we immediately get a good reconstruction
for every \math{k}. The appropriate
sampling probabilities from the previous section are
\mand{
p_t\ge\beta\frac{\a_t\transp\a_t}{\norm{\matA}^2_F}.
}
In this case, if \math{r\ge(4\rho/\beta\epsilon^2)\ln\frac{2d}{\delta}}, then
with probability at least \math{1-\delta},
\mand{
\norm{\matA-\Atilde_k}^2\le\norm{\matA-\matA_k}^2+2\epsilon\norm{\matA}^2.
}
The sampling probabilities are easy to compute and sampling can be accomplished
in one pass if the matrix is stored row-by-row.

To get a relative error result, we need a more carefully constructed
set of  non-uniform sampling
probabilities. The problem here becomes apparent if 
\math{\matA} has rank \math{k}. In this case we have no hope of 
a relative error approximation unless we preserve the rank during sampling. 
To do so, we need to sample according to the actual singular
vectors in \math{\matU}, not according to \math{\matA}; this is because
sampling according to \math{\matA} can give especially large weight to 
a few of the large singular value
directions, ignoring the small singular value directions and hence
not preserving rank.
By sampling according to \math{\matU}, we essentially put equal weight
on all singular directions. To approximate \math{\matU} well, we need 
sampling probabilities
\mand{
p_t\ge\frac{\beta}{d}\u_t\transp\u_t.
}
Then, from Corollary \ref{cor:main1}, if 
\math{r\ge (4(d-\beta)/\beta\epsilon^2)\ln\frac{2d}{\delta}}, with
probability at least \math{1-\delta},
\mand{
\norm{\matI-\matU\transp\matQ\transp\matQ\matU}\le\epsilon.
}
Since \math{\norm{\matU}=1}, it also follows that
\mand{
\norm{\matU\matU\transp-\matU\matU\transp\matQ\transp\matQ\matU\matU\transp}
\le\epsilon.
}
This result is useful because of the following lemma.
\begin{lemma}[\citet{spielman08}]
\label{lemma:basic2} If 
\math{\norm{\matU\matU\transp-
\matU\matU\transp\matQ\transp\matQ\matU\matU\transp}
\le\epsilon}, then for every \math{\x\in\R^d},
\mand{
(1-\epsilon)\x\transp\matA\transp\matA\x
\le
\x\transp\Atilde\transp\Atilde\x
\le
(1+\epsilon)\x\transp\matA\transp\matA\x.
}
\end{lemma}
\begin{proof}
We give a sketch of the proof from \citet{spielman08}.
We let \math{\x\not=\bm0} 
range over \math{\col(\matU)}. Since \math{\col(\matU)=\col(\matA)},
\math{\x\in\col(\matU)} if and only if 
for some \math{\y\in\R^d}, \math{\x=\matA\y}. Since \math{\rank(\matA)=d},
\math{\matA\y\not=0\iff\y\not=0}.
Also note that \math{\matU\matU\transp\matA=\matA}, since
\math{\matU\matU\transp} is a projection operator onto the
column space of \math{\matU}, which is the same as the column
space of \math{\matA}. The following sequence 
establishes the lemma.
\eqan{
\norm{\matU\matU\transp-
\matU\matU\transp\matQ\transp\matQ\matU\matU\transp}
&=&
\sup_{\x\not=\bm0}\frac{|\x\transp\matU\matU\transp\x-
\x\transp\matU\matU\transp\matQ\transp\matQ\matU\matU\transp\x|}
{\x\transp\x},\\
&=&
\sup_{\matA\y\not=\bm0}\frac{|\y\transp\matA\transp
\matU\matU\transp\matA\y-
\y\transp\matA\transp\matU\matU\transp\matQ\transp\matQ\matU\matU\transp
\matA\y|}
{\y\transp\matA\transp\matA\y},\\
&=&
\sup_{\matA\y\not=\bm0}\frac{|\y\transp\matA\transp
\matA\y-
\y\transp\matA\transp\matQ\transp\matQ
\matA\y|}
{\y\transp\matA\transp\matA\y},\\
&=&
\sup_{\y\not=\bm0}\frac{|\y\transp\matA\transp
\matA\y-
\y\transp\Atilde\transp\Atilde\y|}
{\y\transp\matA\transp\matA\y},\\
}
The 
lemma now follows  because
\math{\norm{\matU\matU\transp-
\matU\matU\transp\matQ\transp\matQ\matU\matU\transp}\le\epsilon}.
\end{proof}
Via the Courant-Fischer characterization \cite{golub1983}
of the 
singular values, it is immediate from Lemma~\ref{lemma:basic2}
that the singular value spectrum
is also preserved :
\mld{
(1-\epsilon)\sigma_i(\matA\transp\matA)
\le
\sigma_i(\Atilde\transp\Atilde)
\le
(1+\epsilon)\sigma_i(\matA\transp\matA).
\label{eq:stablesing}
}
Lemma \ref{lemma:basic2} along with \r{eq:stablesing} will allow us
to prove the relative approximation result.
\begin{theorem}
If \math{p_t\ge\frac{\beta}{d}\u_t\transp\u_t} and
\math{r\ge(4(d-\beta)/\beta\epsilon^2)\ln\frac{2d}{\epsilon}}, then,
for \math{k=1,\ldots,d},
\mand{
\norm{\matA-\matA\tilde\Pi_k}\le
\left(\frac{1+\epsilon}{1-\epsilon}\right)^{1/2}
\norm{\matA-\matA_k},
}
where \math{\tilde\Pi_k} projects onto the top \math{k} right singular vectors
of \math{\Atilde}.
\end{theorem}
\paragraph{Remarks} For \math{\epsilon\le\frac12},
 \math{\left(\frac{1+\epsilon}{1-\epsilon}\right)^{1/2}\le 1+2\epsilon}.
Computing the probabilities \math{p_t} involves knowing \math{\u_t} which
means one has to perform an \math{SVD}, in which case, one
could use \math{\matA_k}; it seems like overkill to 
compute \math{\matA_k} in order to approximate \math{\matA_k}. 
We discuss approximate sampling schemes later,
in Section \ref{section:sampling}.
\begin{proof}
Let \math{\norm{\x}=1}. The following sequence establishes the result.
\eqan{
\norm{\matA(\matI-\tilde\Pi_k)}^2
&=&
\sup_{\x\in\ker(\tilde\Pi_k)}\norm{\matA\x}^2=
\sup_{\x\in\ker(\tilde\Pi_k)}\x\transp\matA\transp\matA\x,\\
&\le&
\frac{1}{1-\epsilon}
\sup_{\x\in\ker(\tilde\Pi_k)}\x\transp\Atilde\transp\Atilde\x,\\
&=&
\frac{1}{1-\epsilon}\sigma_{k+1}(\Atilde\transp\Atilde),\\
&\le&
\frac{1+\epsilon}{1-\epsilon}\sigma_{k+1}(\matA\transp\matA)
=\frac{1+\epsilon}{1-\epsilon}\norm{\matA-\matA_k}^2.
}
\end{proof}

\section{\math{\ell_2} Linear Regression with Relative Error Bounds}

A linear regression is represented by a real data matrix 
\math{\A\in\R^{m\times d}} which represents 
\math{m} points in \math{\R^d}, and a target vector \math{\y\in\R^m}.
Traditionally, \math{m\gg d} (severly
over constrained regression).
The goal is to find a regression vector \math{\x^*\in\R^2} which minimizes the
\math{\ell_2} fit error (least squares regression)
\mand{\E(\x)=\norm{\A\x-\y}_2^2=\sum_{t=1}^m(\a_t\transp\x-y_t)^2,}
We assume such an optimal \math{\x^*} exists (it may not be unique
unless \math{\A} has full column rank), and is given by 
\math{\x^*=\A^+\y}, where \math{{}^+} denotes the 
More-Penrose pseudo-inverse;
this problem can be solved in \math{O(md^2)}.
Through row-sampling,
it is 
possible to construct \math{\xhat},
an approximation  to the optimal
regression weights~\math{\x^*}, which is a relative
error approximation to optimal,
\mand{
\E(\xhat)\le(1+\epsilon)\E(\x^*).
}
As usual, let \math{\matA=\matU_\matA\matS_\matA\matV_\matA\transp}.
Then \math{\matA^{+}=\matV_\matA\matS_\matA^{-1}\matU_\matA\transp}, and so 
\math{\x^*=\matV\matS^{-1}\matU\transp\y}. The predictions are
\math{\y^*=\matA\x^*=\matU_\matA\matU_\matA\transp\y}, which 
 is the projection
of \math{\y} onto the column space of \math{\matA}. We define the residual
\math{\bm\epsilon=\y-\y^*=
\y-\matA\x^*=(\matI-\matU_\matA\matU_\matA\transp)\y}, so
\mld{
\y=\matU_\matA\matU_\matA\transp\y+\bm\epsilon.
\label{eq:defepsilon}
}
We will construct
\math{\Atilde} and \math{\tilde\y} by sampling rows:
\mand{
[\Atilde,\tilde\y]=\matQ[\matA,\y],
}
and solve the linear regression problem on 
\math{(\Atilde,\tilde\y)} to obtain
\math{\xhat=\Atilde^{+}\tilde\y}.
For \math{\beta\in(0,\frac13]},
we will use the sampling probabilities 
\mld{
p_t\ge
{\beta}\left(\frac{\u_t^2}{d}+
\frac{
(\u_t^2+
\frac{\epsilon_t^2}{\bm\epsilon\transp\bm\epsilon})}{d+1}+
\frac{\epsilon_t^2}{\bm\epsilon\transp\bm\epsilon}
\right)
\label{eq:LAsamp}
}
to construct \math{\Atilde} and \math{\tilde\y}. There are three parts
to these sampling probabilities. The first part allows us to 
reconstruct \math{\matA} well from \math{\Atilde}; the second allows
us to reconstruct \math{\matA\transp\bm\epsilon}; and, 
the third allows us to reconstruct \math{\bm\epsilon}.

Note that \math{\Atilde=\matQ\matU_\matA\matS_\matA{\matV_\matA}\transp};
if \math{\matQ\matU_\matA} consisted of orthonormal columns, then 
this would be the SVD of \math{\Atilde}. Indeed, this is approximately so,
as we will soon see.
Let 
the  SVD of \math{\Atilde} be \math{\Atilde=\matU_{\Atilde}\matS_{\Atilde}
{\matV_{\Atilde}}\transp}. 
Let \math{\Utilde=\matQ\matU_\matA}.
Since \math{p_t\ge\beta\u_t^2/d}, it follows from
Corollary \ref{cor:main1} that if 
\math{r\ge 2\frac{d-\beta}{\beta\epsilon^2}}, for \math{\epsilon\in(0,1)},
 then, with high probability,
\mand{
\norm{\matI-\Utilde\transp\Utilde}\le\epsilon.
}
Since the eigenvalues of \math{\matI-\Utilde\transp\Utilde} are given by
\math{1-\sigma_i^2(\Utilde)}, it follows that 
\mand{
1-\epsilon<\sigma^2_i(\Utilde)< 1+\epsilon.
}
So all the 
singular values of \math{\matU_\matA} 
are preserved after sampling. Essentially, it 
suffices to sample \math{r=O(d\ln d/\epsilon^2)} rows to preserve
the entire spectrum of \math{\matU_\matA}.
By choosing (say) \math{\epsilon=\frac12},
the rank of \math{\matU_\matA} 
is preserved with high probability, since all 
the singular values are bigger than \math{\frac12}. 
Thus,
with high probability,
\math{\rank(\Atilde)=\rank(\matU_{\Atilde})=\rank(\matQ\matU_{\matA})=
\rank(\matU_{\matA})=\rank(\matA).}
Since \math{\matQ\matU_\matA} 
has full rank, \math{\matS_{\matQ\matU_\matA}^{-1}} is
defined, 
and \math{\matS_{\matQ\matU_\matA}-\matS_{\matQ\matU_\matA}^{-1}} is a 
diagonal matrix whose diagonals are 
\math{(\sigma_i^2(\Utilde)-1)/\sigma_i(\Utilde)};
thus,
\math{\norm{\matS_{\matQ\matU_\matA}-\matS_{\matQ\matU_\matA}^{-1}}_2\le
\epsilon/
\sqrt{1-\epsilon}}. This allows us to quantify the degree to which 
\math{\matQ\matU_\matA} is orthonormal, because
\eqan{
\norm{(\matQ\matU_\matA)^+-(\matQ\matU_\matA)\transp}_2
&=&\norm{\matV_{\matQ\matU_\matA}\matS_{\matQ\matU_\matA}^{-1}
{\matU_{\matQ\matU_\matA}}\transp-
\matV_{\matQ\matU_\matA}\matS_{\matQ\matU_\matA}
\matU_{\matQ\matU_\matA}\transp}_2\nonumber\\
&=&\norm{\matS_{\matQ\matU_\matA}^{-1}-\matS_{\matQ\matU_\matA}}_2
\le\frac{\epsilon}{\sqrt{1-\epsilon}}.\label{eq:nearorth}
}
Finally, we can get a convenient form for \math{\Atilde^+=(\matQ\matA)^+}, 
because
\math{\matQ\matA=\matQ\matU_\matA\matS_\matA\matV_\matA\transp} has full
rank, and so \math{\matQ\matU_\matA=\matU_{\matQ\matU_\matA}
\matS_{\matQ\matU_\matA}\matV_{\matQ\matU_\matA}\transp} has full rank (and
hence is the product of full rank matrices).
Thus,
\eqan{
(\matQ\matA)^+
&=&
(\matU_{\matQ\matU_\matA}
\matS_{\matQ\matU_\matA}\matV_{\matQ\matU_\matA}\transp
\matS_\matA\matV_\matA\transp)^+,
\\
&=&
\matV_\matA(\matS_{\matQ\matU_\matA}\matV_{\matQ\matU_\matA}\transp
\matS_\matA)^+\matU_{\matQ\matU_\matA}\transp,\\
&=&
\matV_\matA\matS_\matA^{-1}
\matV_{\matQ\matU_\matA}\matS_{\matQ\matU_\matA}^{-1}
\matU_{\matQ\matU_\matA}\transp,\\
&=&
\matV_\matA\matS_\matA^{-1}(\matQ\matU_\matA)^+,\\
}
We summarize all this information in the next lemma.
\begin{lemma}
\label{lemma:LA2}
If  \math{r\ge (4d/\beta\epsilon^2)\ln\frac{2d}{\delta}}, 
with probability at least \math{1-\delta}, 
all of the following hold:
\eqar{
&&\rank(\Atilde)=\rank(\matU_{\Atilde})=\rank(\matQ\matU_{\matA})=
\rank(\matU_{\matA})=\rank(\matA);\label{lemma:LR2a}\\
&&\norm{\matS_{\matQ\matU_\matA}-\matS_{\matQ\matU_\matA}^{-1}}_2\le
\epsilon/\sqrt{1-\epsilon};\label{lemma:LR2b}\\
&&\norm{(\matQ\matU_\matA)^+-(\matQ\matU_\matA)\transp}_2
\le \epsilon/\sqrt{1-\epsilon};
\label{lemma:LR2c}\\
&&(\matQ\matA)^+=\matV_\matA\matS_\matA^{-1}(\matQ\matU_\matA)^+.
\label{lemma:LR2d}
}
\end{lemma}
In Lemma \ref{lemma:LA2}
we have simplified the constant to \math{4};
this is a strengthened form of
Lemma 4.1 in \citet{drineas2006}; in particular, the dependence on 
\math{d} is near-linear.

Remember that \math{\hat\x=\Atilde^+\tilde\y}; we now bound
\math{\norm{A\hat\x-\y}^2}. We only sketch the derivation
which basically follows the line of reasoning in \citet{drineas2006}. Under
the conditions of Lemma \ref{lemma:LA2}, with probability at least
\math{1-\delta},
\eqan{
\norm{\matA\hat\x-\y}
&=&
\norm{\matA\Atilde^+\tilde\y-\y}
\ =\ 
\norm{\matA(\matQ\matA)^+\matQ\y-\y}\\
&\buildrel (a)\over =&
\norm{\matU_\matA(\matQ\matU_\matA)^+\matQ\y-\y}\\
&\buildrel (b)\over =&
\norm{\matU_\matA(\matQ\matU_\matA)^+\matQ(\matU_\matA\matU_\matA\transp\y+
\bm\epsilon)-\matU_\matA\matU_\matA\transp\y-\bm\epsilon}\\
&\buildrel (c)\over =&
\norm{\matU_\matA(\matQ\matU_\matA)^+\matQ\bm\epsilon-\bm\epsilon}\\
&=&
\norm{\matU_\matA((\matQ\matU_\matA)^+-(\matQ\matU_\matA)\transp)
\matQ\bm\epsilon+
\matU_\matA(\matQ\matU_\matA)\transp
\matQ\bm\epsilon-\bm\epsilon}\\
&\buildrel (d)\over \le&
\norm{(\matQ\matU_\matA)^+-(\matQ\matU_\matA)\transp}
\norm{\matQ\bm\epsilon}+
\norm{\matU_\matA\transp\matQ\transp
\matQ\bm\epsilon}+\norm{\bm\epsilon}\\
&\buildrel (e)\over \le&
\frac{\epsilon}{\sqrt{1-\epsilon}}
\norm{\matQ\bm\epsilon}+
\norm{\matU_\matA\transp\matQ\transp
\matQ\bm\epsilon}+\norm{\bm\epsilon}.
}
(a) follows from Lemma \ref{lemma:LA2}; (b) follows from \r{eq:defepsilon};
(c) follows Lemma \ref{lemma:LA2}, because \math{\matQ\matU_\matA} has
full rank and so \math{(\matQ\matU_\matA)^+\matQ\matU_\matA=\matI_d};
(d) follows from the triangle inequality and sub-multiplicativity
using \math{\norm{\matU_\matA}=1}; finally, (e) follows from Lemma
\ref{lemma:LA2}. We now see the rationale for the complicated
sampling probabilities. Since
\math{p_t\ge\epsilon_t^2/\bm\epsilon\transp\bm\epsilon}, for \math{r} large 
enough, by Theorem  
\ref{theorem:matmultsym}, \math{\norm{\matQ\bm\epsilon}^2\le
\norm{\bm\epsilon}^2
(1+\epsilon)}. Similarly, since \math{\matU_\matA\transp\bm\epsilon=0},
\math{\norm{\matU_\matA\transp\matQ\transp
\matQ\bm\epsilon}=\norm{\matU_\matA\transp\bm\epsilon-
\matU_\matA\transp\matQ\transp
\matQ\bm\epsilon}}; so,
we can apply Lemma \ref{lemma:main2} with 
\math{\matS_1=\matI_d}, \math{\matV=\bm\epsilon/\norm{\bm\epsilon}} and
\math{\matS_2=\norm{\bm\epsilon}}.
According to Lemma \ref{lemma:main2}, if
\math{p_t\ge \beta(\u_t^2+\epsilon_t^2/\bm\epsilon\transp\bm\epsilon)/(d+1)},
then if \math{r} is large enough, 
\math{\norm{\matU_\matA\transp\matQ\transp
\matQ\bm\epsilon}\le \epsilon\norm{\bm\epsilon}}.
Since these are all probabilistic statements, we need
to apply the union bound to ensure that all of them hold. Ultimately,
we have the claimed result:
\begin{theorem}
\label{theorem:LRrelerr}
For sampling probabilities satisfying \r{eq:LAsamp}, and for  
\math{r\ge(8(d+1)/\beta\epsilon^2)\ln\frac{2(d+1)}{\delta}}, let
\math{\hat\x=(\matQ\matA)^+\matQ\y} be the approximate regression. Then,
with probability at least \math{1-3\delta},
\mand{
\norm{\matA\hat\x-\y}\le \left(1+
\epsilon+\epsilon\sqrt{\frac{1+\epsilon}{1-\epsilon}}
\right)
\norm{\matA\x^*-\y},
}
where \math{\x^*=\matA^+\y} is the optimal regression.
\end{theorem}
\paragraph{Remarks} For the proof of the theorem, we observe that any 
transformation matrix \math{Q} satisfying the following three properties
with high probability will do:
\mand{
(\rn{1})\norm{\matI-\matU\transp\matQ\transp\matQ\matU}\le\epsilon;
\qquad
(\rn{2})
\norm{\matQ\bm\epsilon}\le(1+\epsilon)\norm{\bm\epsilon};
\qquad
(\rn{3})
\norm{
\matU\transp\matQ\transp\matQ\bm\epsilon}\le\epsilon\norm{\bm\epsilon}.
}

\section{Estimating the Spectral Norm}
\label{section:sampling}

The row-norm based sampling is relatively straightforward for the
symmetric product. For the asymmetric product, \math{\matA\transp\matB},
we need probabilities
\mld{
p_t\ge
\beta  \frac{\frac{1}{\norm{\matA}^2}\a_t\transp\a_t+\frac{1}{\norm{\matB}^2}\b_t\transp\b_t}
{\rho_{\matA}+\rho_{\matB}}. \label{eq:ptreq}
}
To get these probabilities, we need \math{\norm{\matA}} and 
\math{\norm{\matB}}; since we can compute the exact product in 
\math{O(md_1d_2)}, a practically useful algorithm would need to estimate
\math{\norm{\matA}} and \math{\norm{\matB}} efficiently.
Suppose we had estimates \math{\lambda_\matA,\lambda_\matB} which satisfy:
\mand{
(1-\epsilon)\norm{\matA}^2\le\lambda_\matA^2\le(1+\epsilon)\norm{\matA}^2;
\qquad
(1-\epsilon)\norm{\matB}^2\le\lambda_\matB^2\le(1+\epsilon)\norm{\matB}^2.
}
We can construct probabilities satisfying the desired property
with 
\math{\beta=(1-\epsilon)/(1+\epsilon)}.
\eqan{
p_t
&=&
\frac{\frac{1}{\lambda_\matA^2}\a_t\transp\a_t+\frac{1}{\lambda_\matB^2}\b_t\transp\b_t}
{{\norm{\matA}^2_F}/{\lambda_\matA^2}+
{\norm{\matB}^2_F}/{\lambda_\matB^2}}\\
&\ge&
\frac{\frac{1}{(1+\epsilon)\norm{\matA}^2}\a_t\transp\a_t+\frac{1}{(1+\epsilon)\norm{\matA}^2}\b_t\transp\b_t}
{{\norm{\matA}^2_F}/{(1-\epsilon)\norm{\matA}^2}+
{\norm{\matB}^2_F}/{(1-\epsilon)\norm{\matA}^2}}\\
&=&
\left(\frac{1-\epsilon}{1+\epsilon}\right)
\frac{\frac{1}{\norm{\matA}^2}\a_t\transp\a_t+\frac{1}{\norm{\matB}^2}\b_t\transp\b_t}
{\rho_{\matA}+\rho_{\matB}}.
}
One practical way to obtain \math{\norm{\matA}^2} is using the power iteration.
Given an arbitrary unit vector \math{\x_0}, 
for \math{n\ge1}, let 
\math{\x_n=\matA\transp\matA\x_{n-1}/\norm{\matA\transp\matA\x_{n-1}}}.
Note that multiplying by \math{\matA\transp\matA} can be done in 
\math{O(2md_1)} operations. Since \math{\x_n} is a 
unit vector,
\math{\norm{\matA\transp\matA\x_{n}}\le\norm{\matA}^2}. We now get a lower 
bound.
Let \math{\x_0} be a random isotropic vector constructed using 
\math{d_1} independent standard Normal variates \math{z_1,\ldots,z_{d_1}};
so \math{\x_0\transp=[z_1,\ldots,z_{d_1}]/\sqrt{z_1^2+\cdots+z_{d_1}^2}}.
Let \math{\lambda_n^2=\norm{\matA\transp\matA\x_{n}}} be an estimate
for \math{\norm{\matA}^2} after \math{n} power iterations.
\begin{lemma}
\label{lemma:estimatenorm}
For some constant \math{c\le(\frac{2}{\pi}+2)^3},
 with probability at least \math{1-\delta},
\mand{
\lambda_n^2\ge
\frac{\norm{\matA}^2}{\sqrt{4+\frac{cd_1}{\delta^3}\cdot 2^{-2n}}}.
} 
\end{lemma}
\paragraph{Remarks}
\math{n\ge c \log\frac{d_1}{\delta}}
gives the desired constant factor approximation.
Since each power iteration takes \math{O(md_1)} time, and we run 
\math{O(\log\frac{d_1}{\delta})} power iterations, 
in \math{O(md_1\log\frac{d_1}{\delta})} time, we obtain a
sufficiently good estimate for \math{\norm{\matA}} (and similarly for 
\math{\norm{\matB}}).
\begin{proof}
Assume that \math{\x_0=\sum_{i=1}^{d_1}\alpha_i\v_i}, where \math{\v_i} are the
eigenvectors of \math{\matA\transp\matA} with corresponding 
eigenvalues
\math{\sigma_1^2\ge\cdots\ge\sigma_{d_1}^2}. 
Note \math{\norm{\matA}^2=\sigma_1^2}.
If \math{\sigma_{d_1}^2\ge{\sigma_1^2/2}}, then it trivially follows that
\math{\norm{\matA\transp\matA\x_n}\ge\sigma_1^2/2} for any \math{n}, so 
assume that 
\math{\sigma_{d_1}^2<{\sigma_1^2/2}}.  We can thus partition the
singular values into those at least \math{\sigma_1^2/2} and those which 
are smaller; the latter set is non-empty.
So assume for some \math{k<d_1}, \math{\sigma_k^2\ge\sigma_1^2/2} and
\math{\sigma_{k+1}^2<\sigma_1^2/2}. 
Since \math{\x_n=\sum_{i}\alpha_i\sigma_i^{2n\v_i}/(\sum_{i}
\alpha_i^2\sigma_i^{4n})^{1/2}},
we therefore have:
\eqan{
\lambda_n^4
&=&
\norm{\matA\transp\matA\x_n}^2=
\frac{\sum_{i=1}^{d_1}\alpha_i^2\sigma_i^{4(n+1)}}
{\sum_{i=1}^{d_1}\alpha_i^2\sigma_i^{4n}}\\
&\ge&
\frac{\sum_{i=1}^{k}\alpha_i^2\sigma_i^{4(n+1)}}
{\sum_{i=1}^{d_1}\alpha_i^2\sigma_i^{4n}}\\
&=&
\frac{\sum_{i=1}^{k}\alpha_i^2\sigma_i^{4(n+1)}}
{\sum_{i=1}^{k}\alpha_i^2\sigma_i^{4n}+\sum_{i=k+1}^{d_1}
\alpha_i^2\sigma_i^{4n}},\\
&=&
\sigma_1^4
\frac{\sum_{i=1}^{k}\alpha_i^2(\sigma_i/\sigma_1)^{4(n+1)}}
{\sum_{i=1}^{k}\alpha_i^2(\sigma_i/\sigma_1)^{4n}+\sum_{i=k+1}^{d_1}
\alpha_i^2(\sigma_i/\sigma_1)^{4n}},\\
&\buildrel(a)\over\ge&
\sigma_1^4
\frac{\sum_{i=1}^{k}\alpha_i^2(\sigma_i/\sigma_1)^{4(n+1)}}
{4\sum_{i=1}^{k}\alpha_i^2(\sigma_i/\sigma_1)^{4(n+1)}+2^{-2n}},\\
&=&
\frac{\sigma_1^4}
{4+2^{-2n}/\sum_{i=1}^{k}\alpha_i^2(\sigma_i/\sigma_1)^{4(n+1)}},\\
&\buildrel(b)\over\ge&
\frac{\sigma_1^4}
{4+2^{-2n}/\alpha_1^2}.
}
(a) follows because for \math{i\ge k+1},
\math{\sigma_i^2<\sigma_1^2/2}; for \math{i\le k},
\math{\sigma_1^2/\sigma_i^2\le 4}; and \math{\sum_{i\ge k+1}\alpha_i^2\le
\sum_{i\ge 1}\alpha_i^2=1}. (b) follows because 
\math{
\sum_{i=1}^{k}\alpha_i^2(\sigma_i/\sigma_1)^{4(n+1)}\ge \alpha_1^2}.
The theorem will now follow if we show that with probability 
at least \math{1-c\delta^{1/3}},
\math{\alpha_1^2\ge \delta/d}.
It is clear that \math{\Exp[\alpha_1^2]=1/d} from isotropy.
Without loss of generality, assume \math{\v_1} is aligned with the 
\math{z_1} axis. So 
\math{\alpha_1^2=z_1^2/\sum_iz_i^2}
(\math{z_1,\ldots,z_d} are independent standard normals).
For \math{\delta<1},
we estimate \math{\Prob[\alpha_1^2\ge\delta/d]} as follows:
\eqan{
\Prob\left[\alpha_1^2\ge\frac{\delta}{d}\right]
&=&\Prob\left[\frac{z_1^2}{\sum_iz_i^2}\ge\frac{\delta}{d}\right]
\ =\ \Prob\left[{z_1^2}\ge\frac{\delta}{d}{\sum_iz_i^2}\right]
\ =\ \Prob\left[{z_1^2}\ge\frac\delta{d-\delta}{\sum_{i\ge2}z_i^2}\right]\\
&\ge&\Prob\left[{z_1^2}\ge\frac\delta{d-1}{\sum_{i\ge2}z_i^2}\right]\\
&\buildrel(a)\over=&\Prob\left[\chi^2_1\ge\frac\delta{d-1}\chi^2_{d-1}\right],\\
&\buildrel(b)\over\ge&\Prob\left[\chi^2_1\ge \delta+\delta^{2/3}\right]
\cdot\Prob\left[
\frac\delta{d-1}\chi^2_{d-1}\le \delta+\delta^{2/3}\right].\\
}
In (a) we compute the probability that a \math{\chi^2_1} random variable
exceeds a multiple of an independent \math{\chi^2_{d-1}} random variable, which
follows from the definition of the \math{\chi^2} distribution 
as a sum of squares
of 
independent standard normals.
(b) follows from independence and 
because one particular realization of the event in (a) is when
\math{\chi^2_1\ge \delta+\delta^{2/3}} and 
\math{\delta\chi^2_{d-1}/(d-1)\le \delta+\delta^{2/3}}.
Since \math{\Exp[\chi^2_{d-1}/(d-1)]=1}, and 
\math{Var[\chi^2_{d-1}/(d-1)]=2/(d-1)}, by Chebyshev's inequality,
\mand{
\Prob\left[
\frac\delta{d-1}\chi^2_{d-1}\le \delta+\delta^{2/3}\right]
\ge 1-\frac{2\delta^{1/3}}{d-1}.
}
From the definition of the \math{\chi^2_1} distribution, we can bound
\math{\Prob[\chi^2_1\le\delta+\delta^{2/3}]},
\mand{
\Prob[\chi^2_1\le\delta+\delta^{2/3}]
=
\frac{1}{2^{1/2}\Gamma(1/2)}
\int_{0}^{\delta+\delta^{2/3}}du\ u^{-1/2}e^{-u/2}
\le \sqrt{\frac{2}{\pi}}(\delta+\delta^{2/3})^{1/2},
}
and so 
\mand{
\Prob\left[\alpha_1^2\ge\frac{\delta}{d}\right]
\ge
\left(1-\sqrt{\frac{2}{\pi}}(\delta+\delta^{2/3})^{1/2}\right)
\cdot
\left(1-\frac{2\delta^{1/3}}{d-1}\right)
\ge1-\left(\frac2\pi+2\right)\delta^{1/3}
.
}
\end{proof}

We now consider the sampling based approach to estimate
the spectral norm. Pre-sample 
the rows of \math{\matA} using probabilities proportional to the
row norms to construct \math{\Atilde}. We know that if 
\math{r\ge(4\rho_A/\beta\epsilon^2)\ln\frac{2d_1}{\delta}}, then
\mand{
\norm{\Atilde\transp\Atilde-\matA\transp\matA}\le
\epsilon\norm{\matA}^2.
}
It follows that we have a \math{\epsilon}-approximation to the 
spectral norm from
\eqan{
\norm{\Atilde\transp\Atilde}
&=&
\norm{\Atilde\transp\Atilde-\matA\transp\matA+\matA\transp\matA}
\ \le\ 
(1+\epsilon)\norm{\matA}^2;\\
\norm{\matA\transp\matA}
&=&
\norm{\matA\transp\matA-\Atilde\transp\Atilde+\Atilde\transp\Atilde}
\ \le\ 
\epsilon\norm{\matA}^2+\norm{\Atilde\transp\Atilde}.
}
Thus,
\math{(1-\epsilon)\norm{\matA}^2\le\norm{\Atilde\transp\Atilde}
\le(1+\epsilon)\norm{\matA}^2}. Along this route, one must 
first sample 
\math{r} rows, and then approximate the spectral norm of the resulting
\math{\Atilde}. We may now combine with the power iteration  on 
\math{\Atilde\transp\Atilde} to get a 
constant factor approximation efficiently (or we may 
compute exactly in \math{O(rd_1^2)}). Specifically, set 
\math{\epsilon=\frac12}, in which case, with high probability, 
\math{\frac12\norm{\matA}^2\le\norm{\Atilde\transp\Atilde}
\le\frac32\norm{\matA}^2}. Now, choose the number of power iterations 
\math{n\ge n^*}, where 
\math{\frac{cd_1}{\delta^3}=2^{n^*}}. In this case, after 
\math{n}  power iterations, we have an estimate which is at least 
\math{\frac{1}{\sqrt{5}}\norm{\Atilde}^2} from Lemma \ref{lemma:estimatenorm},
which proves Theorem~\ref{theorem:spectral}.
\begin{theorem}
\label{theorem:spectral}
With \math{r\ge(4\rho_A/\epsilon^2)\ln\frac{2d_1}{\delta}}, 
the spectral norm estimate \math{\tilde\sigma_1^2}
obtained
after \math{c\ln\frac{d_1}{\delta}} power 
iterations on \math{\Atilde\transp\Atilde} starting from an
isotropic random vector satisfies
\mand{
\frac{1}{2\sqrt{5}}\norm{\matA}^2\le\tilde\sigma_1^2\le\frac{3}{2}
\norm{\matA}^2.
}
Further, the estimate \math{\tilde\sigma_1^2} can be computed in
\math{O(md_1+\rho_A d_1/\epsilon^2\ln^2 (\frac{d_1}{\delta}))}.
\end{theorem}
As mentioned at the begining of this section, 
constant factor approximations to the spectral norms
 of the relevant matrices is enough to obtain probabilities satisfying
\r{eq:ptreq} for some constant \math{\beta}.

{
\bibliographystyle{natbib}
\bibliography{active,mypapers,masterbib} 
}


\end{document}